\documentclass[12pt]{article}
\pdfoutput=1
\usepackage{amsmath,amssymb, braket, amsthm, bm}
\usepackage[a4paper]{geometry}
\usepackage[colorlinks=true, pdfstartview=FitV, linkcolor=blue, citecolor=blue, urlcolor=blue]{hyperref}
\usepackage{graphicx}
\usepackage{subcaption}
\usepackage{xcolor}

\newcommand{\ind}{\mbox{ ind}}
\def\ran{\mathop{\rm Ran}\nolimits}
\newcommand {\bB}{{\mathbb B}}

\newcommand {\bC}{{\mathbb C}}
\newcommand {\bZ}{{\mathbb Z}}
\newcommand {\bN}{{\mathbb N}}
\newcommand {\bM}{{\mathbb M}}
\newcommand {\bR}{{\mathbb R}}
\newcommand {\bS}{{\mathbb S}}
\newcommand {\bH}{{\mathbb H}}
\newcommand {\bL}{{\mathbb L}}
\newcommand {\bI}{{\mathbb I}}

\newcommand{\cO}{{\mathcal O}}

\newcommand{\cshift}{{\bf T}}
\newcommand {\flux}{{\Phi}}
\newcommand {\bbM}{{\bf M}}
\newcommand {\bbC}{{\bf C}}
\newcommand {\bbP}{{\bf P}}

\newcommand{\be}{\begin{equation}}
\newcommand{\ee}{\end{equation}}

\newtheorem{thm}{Theorem} [section]
\newtheorem{lem}[thm]{Lemma}
\newtheorem{prop}[thm]{Proposition}
\newtheorem{proposition}[thm]{Proposition}

\newtheorem{definition}[thm]{Definition}
\newtheorem{cor}[thm]{Corollary}

\newtheorem {rem}[thm]{Remark}
\newtheorem {rems}[thm]{Remarks}

\newtheorem {example}[thm]{Example}
\newtheorem {examples}[thm]{Examples}

\title{ Engineering  stable quantum currents at bulk boundaries
\thanks{ Supported by
FONDECYT 1161732, 
and ECOS-Conicyt C15E10}\  \thanks{Supported by the LabEx PERSYVAL-Lab (ANR-11- LABEX- 0025-01) funded by the French program Investissement d'Avenir and ANR Grant NONSTOPS (ANR-17-CE40-0006-01)}}
\author{Joachim Asch  \thanks{CNRS, CPT, Aix Marseille Universit\'e, Universit\'e de Toulon, Marseille, France, asch@cpt.univ-mrs.fr},
Olivier Bourget
\thanks{
Departamento de Matem\'aticas
Pontificia Universidad Cat\'olica de Chile, Av. Vicu\~{n}a Mackenna 4860,
C.P. 690 44 11, Macul
Santiago, Chile},
Alain Joye
\thanks{
Universit\'e Grenoble Alpes, CNRS Institut Fourier, 38000 Grenoble, France}
}
\date{19/6/19}
\begin{document}
\maketitle

\begin{abstract}

We study transport properties of discrete quantum dynamical systems on the lattice, in particular Coined Quantum Walks and the Chalker--Coddington  model. We prove  existence of a non trivial charge transport and that the absolutely  continuous spectrum  covers the whole unit circle under mild assumptions. For Quantum Walks  we exhibit explicit constructions of coins which imply existence of stable directed quantum currents along classical curves. The results are of topological nature and independent of the details of the model. \end{abstract}

\section{Introduction}

We consider the signature of transport provided by the presence of a non trivial absolutely continuous component in the spectrum of the evolution operator of two classes of unitary network models: the celebrated Chalker-Coddington model of condensed matter physics, \cite{cc,kok}, which provides an effective description of one time step of  the motion of an electron in a plane subject to a strong perpendicular magnetic field and a random potential, and the abstract $d$-dimensional  Coined Quantum Walk which finds applications in several areas of Quantum Computing, \cite{VA,P}. Both models are characterized by a unitary operator $U$ on a Hilbert space $\ell^2(\bZ^d,\bC^n)$ which couples neighboring sites only. 

Moreover both classes are parametrised by a countable family of unitary finite-dimensional matrices that describe the local dynamics: the scattering matrices for the Chalker-Coddington model and the coin matrices for the Quantum Walks. This makes these models very versatile and they display a whole range of quantum dynamics depending on the choice of these matrices. In particular, it is known that dynamical Anderson localisation takes place for random versions of the Chalker-Coddington model, \cite{ABJ1,ABJ2}, and of Coined Quantum Walks, \cite{JM,ASW,J}. On the other hand  for Quantum Walks on trees there are localisation-delocalisation transitions,\cite{HJ} and the localisation length diverges with high coordination number  for  Balanced Random Quantum Walks, \cite{AJ}. Moreover, homogeneous parameters at infinity or the presence of boundaries induce absolutely continuous spectrum for boths models, \cite{JMa, ABJ3,ABJ4}.  For periodically driven unitary networks models it is known that the presence of boundaries and symmetries of the bulk imply occurence of currents, \cite{rlbl,dft,graftauber,ssb}.

By contrast, here we analyse transport properties of  models defined on the whole of $\bZ^d$.

To do so we use the topological properties of the self-adjoint flux operator out of a subspace $\ran(P)$ of an orthogonal projection
$$
\flux=U^*PU-P.
$$ 
 It is known that the index of $\Phi$ is invariant under small or compact variations of $U$ or $P$.  We prove in general that a non trivial index implies the appearance of a wandering subspace reducing a perturbation of the  evolution operator which becomes  a shift on the subspace, see Theorem \ref{thm:general}. This may be considered as a current and implies the occurence of gapless absolutely continuous spectrum under decay assumptions. As a by-product we prove a trace formula for even powers of $\Phi$ for its index, Proposition {\ref{lem:wandering}}.
 
 Then we turn to Coined Quantum Walks and analyse the properties of the flux operator as a function of the local coin matrices for $P$ being a full bulk state on a half space with boundary. We show how to engineer currents along  leads on an interface by choosing coin matrices, Theorem \ref{thm:quantumleads}. These currents imply existence of absolutely continuous spectrum stable under  perturbations at all quasienergies.

The last Section is devoted to the Chalker-Coddington model defined on the plane. We show that under very general conditions there exists a non trivial flux of particles through a curve. Under mild additional assumptions it implies  gapless absolutely continuous spectrum, see Theorem \ref{thm:main} and remark \ref{rem:cap}. 

\section{General results}
In \cite{ASS} the  relative index of  two projections was defined and shown to be zero if and only if there exists a unitary operator interchanging the projections. A trace formula in odd powers of their difference was shown for the index. In this section we shall prove spectral properties of $U$ which are implied by the non vanishing relative index of $U^\ast P U$ and $P$. We also observe that a trace formula in even powers of $\Phi$ holds for the relative index. This is a quite straightforward implication of the work of \cite{ASS} and quite useful in our application in chapter \ref{section:cc}. The occurence of a wandering subspace for $U$ was observed by \cite{cetal}; item (2) of the following Theorem is their Proposition 7.1. We proved item (3) for a special case in \cite{ABJ4}.
\begin{thm}\label{thm:general}
Let $U$ be a unitary operator on a Hilbert space and $P$ an orthogonal projection. For the selfadjoint operator
\[\Phi:=U^\ast P U-P=U^\ast\left[P,U\right] \] 
suppose that  $1$ is an isolated eigenvalue of finite multiplicity of  $\Phi^2$ and define the integer
\[\ind(\Phi)
:= \dim\ker\left(\Phi-\bI\right)-\dim\ker\left(\Phi+\bI\right).\]
If the index does not vanish,  $\ind(\Phi)=n\neq0$, then:

there exists a unitary $\widehat{U}$ such that   $\widehat{U}=S\oplus\widetilde{U}$ where $S$ is a bilateral shift of multiplicity $\vert n\vert$, $\widetilde{U}$ is unitary on its subspace , $\lbrack\widetilde{U},P\rbrack=0$ and :
\begin{enumerate}
\item $\Vert U-\widehat{U}-F\Vert=\cO(\Vert\Phi_<\Vert)$ for a finite rank operator $F$  and $\Phi_<$ the restriction of $\Phi$ to its spectral subspace off $\pm1$ : $\Phi_<:=\Phi\chi(\Phi^2<1)$;
\item
if  $\lbrack P,U\rbrack$ is compact then $U-\widehat{U}$ is compact   and  the essential spectrum of $U$ is the whole unit circle :
\[\sigma(U)=S^1;\]
\item if  $\lbrack P,U\rbrack$ is trace class then $U-\widehat{U}$ is trace class   and the absolutely continuous spectrum of $U$ is the whole unit circle:
\[\sigma_{ac}(U)=S^1.\]
\end{enumerate}\end{thm}

For the proof and for our applications we shall use the following facts and concepts. With the notation $P^\perp:=1-P$ :

\begin{proposition}{\label{lem:wandering}}
\begin{enumerate}
\item \[\Phi^2\le1; \quad\lbrack\Phi^2,P\rbrack=0;\]
\[ \ker\left(\Phi+\bI\right)=\ker\left(PUP\restriction\ran{P} \right); \quad\ker\left(\Phi-\bI \right)=\ker\left( P^\perp U P^\perp\restriction\ran{P^\perp}\right).\]
\item \[ind(\Phi)=trace\left(\Phi^{2j+1} \right)\]
for any odd integer $2j+1$ for which $\Phi^{2j+1}$ is trace class.
\item If $\ind(\Phi)$ is defined then  \[\ind(\Phi)= \dim\ker\left((\Phi^2-\bI)\restriction\ran{P^\perp}\right)-\dim\ker\left((\Phi^2-\bI)\restriction\ran{P}\right).\]
Furthermore:
 \[ind(\Phi)=trace\left((P^\perp-P)\Phi^{2j}\right)\]
for any even integer $2j$ for which $\Phi^{2j}$ is trace class; in particular if $\Phi^2$ is trace class then
\[ind(\Phi)=trace\left(P^\perp U^\ast P U P^\perp-PU^\ast P^\perp U P\right).\]
\item If $\lbrack0,1\rbrack\ni t\to U(t)$ is norm continuous and unitary and for $\Phi(t)=U^\ast(t)PU(t)$: $1\notin\sigma_{ess}(\Phi(t)^2)$ then $\bZ\ni\ind(\Phi(t))=const.$ If $\Phi(t_0)^j$ is trace class for a $j\in\bN$ then $\ind(\phi(t))$ can be calculated by a trace as in item (2) or (3).

\item  A $d$-dimensional subspace $\bL$ is called wandering  for a unitary $U$, if ${{U}}^k\bL\perp\bL\quad \forall k\in\bN$. 
For an orthogonal decomposition $\bL=\bigoplus_{j=1}^{d}\bL_j$ into 1-dimensional subspaces  and for the $U$-invariant subspace
\[\bM:=\bigoplus_{k\in\bZ} {{U}}^k\bL=\bigoplus_{j=1}^{d} \bigoplus_{k\in\bZ}{{U}}^k\bL_j\]
it holds that $S:={U}\restriction\bM$ is a bilateral shift of multiplicity  $d$   and $U\restriction\bM^\perp$ is unitary on $\bM^\perp$. In particular $\sigma_{ac}(U)=\bS^1$.
\end{enumerate}
\end{proposition}

\begin{rem}
\item The assertions $1,2,5$ are well known, see \cite{ASS, SNF}. We contribute  the supertrace formula for even powers of $\Phi$ for the index, i.e.  item  (3). This  is in fact a generalization of Kitaev's formula \cite{ki}, which says for  the case of $P$ a multiplication operator in $\ell^2(\bZ;\bC)$:
\[\ind(\Phi)=trace\left((P^\perp-P)\Phi^{2}\right)=\sum_{x\in P, y\in P^\perp}\vert U_{xy}\vert^2-\vert U_{yx}\vert^2.\]
We have formulated the result $(3)$ for $\Phi=U^\ast PU -P$, but,  as for the result $(2)$,  see \cite{ASS}, the proof works in the general case where $\Phi$  is the difference of  two projections $\Phi=R-P$.
\end{rem}
\begin{proof} of Proposition ({\ref{lem:wandering}}). Item(4) follows from items (2,3) and the stability of the index, c.f. \cite{ASS}. \\
Item (3): $\psi\in\ker(\Phi+\bI)\Longleftrightarrow(\psi\in\ker(\Phi^2-\bI)\hbox{ and }\psi=P\psi$);   $\psi\in\ker(\Phi-\bI)\Longleftrightarrow(\psi\in\ker(\Phi^2-\bI)\hbox{ and }\psi=P^\perp\psi)$ which implies the first assertion. 

Denote the projection $R:=U^\ast P U$ and $B:=R^\perp-P$. Then $\Phi=R-P$ and
\[\Phi+B=P^\perp-P.\]
Denote $Q_{\lambda^2}$ the spectral projection on $ker(\Phi^2-\lambda^2)$ whose dimension is finite, $\Phi^2$ being compact. From  \cite{ASS} we know  that $\Phi^2+B^2=\bI$ and $\{\Phi,B\}=0$ and by their Theorem 4.2 that: 
\[trace\left(Q_{\lambda^2} (\Phi+B)\right)=0\hbox{ if } 0<\lambda^2<1.\]
Assuming that $\Phi^{2j}$ is trace class and calculating the trace in the spectral decomposition of  $\Phi^2$ we have
\[trace\left((P^\perp-P)\Phi^{2j}\right)=\sum_{\lambda^2\in\sigma(\Phi^2),\lambda^2<1}\lambda^{2j}\underbrace{trace\left(Q_{\lambda^2} (\Phi+B)\right)}_0+trace\left(Q_{1} (\Phi+B)\right)\]
\[=trace\left(Q_{1} (\Phi+B)\right)=trace\left(Q_{1} \Phi\right)=\dim\ker\left(\Phi-\bI\right)-\dim\ker\left(\Phi+\bI\right).\]
\end{proof}
\begin{proof} of Theorem \ref{thm:general}
Following the idea of \cite{cetal} we show that there exists a  perturbation of $U$ which admits an $\vert n\vert$ dimensional wandering subspace $\bL$. We freely make use of Proposition (\ref{lem:wandering}). 

For the dimensions of the $\pm1$ eigenspaces of $\Phi$ we can suppose $n=n_+-n_- < 0$ (else consider $P^\perp$  in place of $P$ as\  $-\Phi=U^\ast P^\perp U-P^\perp$). Then
\[\dim\ker\left(\Phi+\bI\right)=\dim\ker\left(\Phi-\bI\right)+\vert n\vert= n_+ + \vert n\vert.\]
Choose an $\vert n\vert$  dimensional subspace
\[\bL\subset \ker\left(\Phi+\bI\right)=\ker\left(PUP\restriction\ran{P} \right)\]
and denote its $n_+$ dimensional complement  $\bL^\perp:=\ker\left(\Phi+\bI\right)\ominus\bL\subset\ran P$. The orthogonal projection $Q_-$ on $\ker\left(\Phi+\bI\right)$ then decomposes as $Q_-=Q_\bL+Q_{\bL^\perp}$.  $Q_+$ projects on $\ker\left(\Phi-\bI\right)\subset\ran P^\perp$  and $Q_0:=\bI-Q_+-Q_-$. Note that $Q_0=\chi\left(\Phi^2\in[0,1)\right)$, the spectral projection off $1$.

Choose a $n_+$ dimensional unitary $V$ 

\[V:\bL^\perp\to\ker\left(\Phi-\bI\right)=\ran Q_+=\ker\left(P^\perp UP^\perp\restriction\ran{P^\perp} \right).\]

Observe for the diagonal operator :
\[PUP+P^\perp U P^\perp=U-\lbrack P,\lbrack P,U\rbrack\rbrack=U ( \underbrace{\bI-\Phi^2-\lbrack P,\Phi\rbrack}_{=:W} ).\]
Denoting the projection $R:=U^\ast P U$, it holds : $\Phi=R-P$, $\bI-\Phi^2= P^\perp R^\perp+RP$,  $W=R^\perp P^\perp+RP$. It follows 
\[W^\ast W=P^\perp R^\perp P^\perp+PRP=\bI-\Phi^2.\]
$W$ commutes with $\Phi^2$ thus with $Q_0$ because $P$ does. Denote  $U_0:\ran Q_0\to\ran Q_0$ the unitary such that 
\[W=U_0 (\bI-\Phi^2)^{1/2} \hbox{ on  }\ran Q_0\]
and  $\widehat{U}$ the unitary on the entire Hilbert space :
\[\widehat{U}\left(\psi_\bL+\psi_-+\psi_+ +\psi_0\right):= U\left(\psi_\bL+V^\ast\psi_+ + V\psi_-+U_0\psi_0\right)\]
for $\psi_\bL+\psi_-+\psi_+ +\psi_0 \in \ran Q_\bL+\ran Q_\bL^\perp+\ran Q_+ +\ran Q_0$.
%

By construction we have (see also figure \ref{fig:widehatU}): 
\begin{eqnarray*}
&&\psi_\bL\in\ker PUP\Rightarrow \psi_\bL\in\ker P\widehat{U}P,\\
&&\psi_-\in\ker PUP\Rightarrow V\psi_-\in\ker P^\perp UP^\perp \Rightarrow\widehat{U}\psi_-=\widehat{U}P\psi_-=P\widehat{U}P\psi_-,\\
&&\psi_+\in\ker P^\perp UP^\perp\Rightarrow V^\ast\psi_+\in \ker PUP\Rightarrow\widehat{U}\psi_+=\widehat{U}P^\perp\psi_+=P^\perp\widehat{U}P^\perp\psi_+,\\
&&\widehat{U}Q_0=(PUP+P^\perp U P^\perp)\left(\bI-\Phi^2\right)^{-1/2}Q_0=(P\widehat{U}P+P^\perp\widehat{U}P^\perp)Q_0\\
&&\hbox{ as } \lbrack P,\Phi^2\rbrack=0.
\end{eqnarray*}

It follows  \[\ker\left( P \widehat{U} P\restriction\ran P\right)=\bL; \quad \ker\left( P^\perp \widehat{U} P^\perp\restriction\ran P^\perp\right)=0; \quad P\widehat{U} P^\perp=0.\]
This implies for all $k\in\bN$
\begin{figure}[hbt]
\centerline {
\includegraphics[width=9cm,height=4.5cm,keepaspectratio=false]{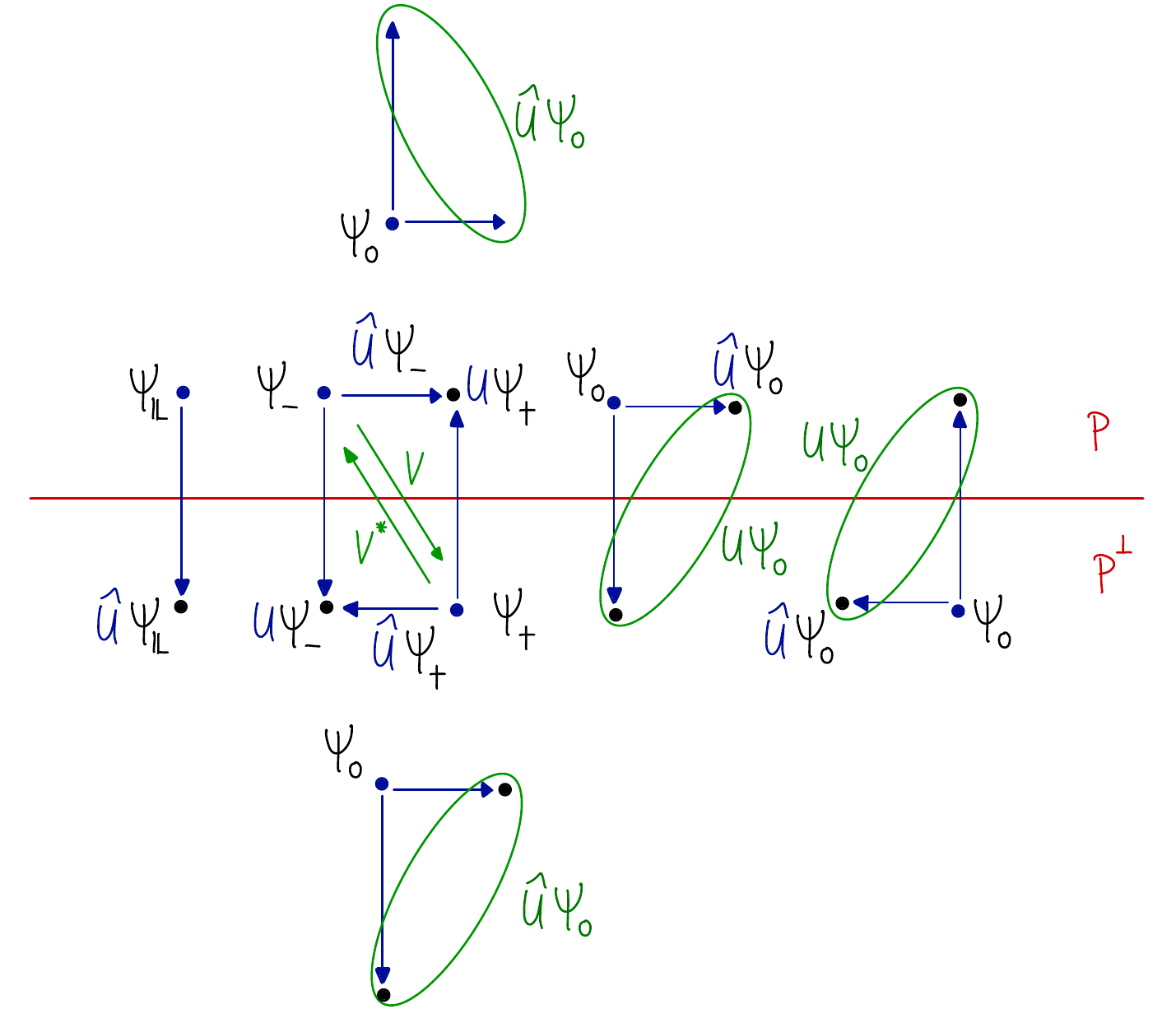}
}
\caption{The action of $\widehat{U}$}
\label{fig:widehatU}
\end{figure}

%

\[P\widehat{U}^k P^\perp=0,\quad P\widehat{U}^kQ_\bL=0,  \quad \hbox{ and in particular} \quad Q_\bL\widehat{U}^k Q_\bL=0,\]
thus  $\bL$ is a wandering subspace for $\widehat{U}$.  This implies  for the invariant subspace $\bM:=\bigoplus_\bZ {\widehat{U}}^k\bL$ that the restriction
$S:=\widehat{U}\restriction\bM$  is a bilateral shift of multiplicity  $\vert n\vert$ and that  $\widetilde{U}:=\widehat{U}\restriction\bM^\perp$ is unitary.

Remark that $\widehat{U}(\bI-Q_\bL)$ commutes  with $P$ and in particular $\lbrack Q_{\bM},P\rbrack=0$ which implies $\lbrack\widetilde{U},P\rbrack=0$.

%
%

Concerning the regularity remark that 
\[U^\ast \widehat{U}-\bI=\underbrace{(V^\ast-\bI)Q_+ +(V-\bI)Q_{\bL^\perp}}_{=:F}+(U_0-\bI)Q_0\] 

 where $F$ is of finite rank, and,  as operators on $\ran Q_0$ :
\[(U_0-\bI)(\bI-\Phi^2)^{1/2}=W-(\bI-\Phi^2)^{1/2}=-\Phi^2-\lbrack P,\Phi\rbrack +\Phi^2(\bI+(\bI-\Phi^2)^{1/2})^{-1}\]
which is of order $\cO(\Vert\Phi\chi(\Phi^2<1)\Vert)$ and compact (trace class)  if $\Phi$ is compact (trace class).

The results on the spectrum of $U$ then follow from invariance of the essential (resp. ac) spectrum under compact  (resp. trace class) perturbations.%
\end{proof}

\section{Engineering Quantum Walks with prescribed currents on bulk boundaries}

We consider coined quantum walks on $\bH=\ell^2\left(\bZ^d,\bC^{2d}\right)$, with scalar product $\langle\cdot,\cdot\rangle$. Denote the discrete unit sphere by 
\[I_{2d}:=S^{d-1}\cap\bZ^d\]
 and its directions by $\left(\pm j\right):=\pm e_j$, $j\in\{1,\ldots,d\}$, with  the canonical basis vectors $e_j\in\bZ^d$; thus $I_{2d}=\left\{\left(\pm j\right) \right\}$. A quantum direction is a basis vector $\Ket{\pm j}\in\bC^{2d}$ with $\Ket{+j}:=\widehat{e}_{2j-1}$, $\Ket{-j}:=\widehat{e}_{2j}$ for the canonical basis vectors $\widehat{e}_j\in\bC^{2d}$; we freely use the bijection from classical to quantum directions : $\left(\pm j\right)\leftrightarrow\Ket{\pm j}$.

For $\tau\in I_{2d}$ denote 
\[P_\tau:=\Ket\tau\Bra\tau\]
 the orthogonal projection on $span(\Ket{\tau})\subset\bC^{2d}$.
 
 By a multiplication operator  on $\bH$ with symbol $M$ we understand
\[\bbM\psi(x)=M(x)\psi(x), \quad \psi\in\bH, \quad M(x)\in\bB\left(\bC^{2d}\right).\]

 Denote $S_\tau$ the right shift $S_\tau\psi(x):=\psi(x-\tau)$. The conditional shift operator is 
\begin{equation}\label{def:cshift}
\cshift:\bH\to\bH, \quad \cshift:=\sum_{\tau\in I_{2d} } S_\tau \bbP_\tau,\end{equation}
where $\bbP_\tau$ has the symbol $P_\tau,\forall x$, i.e.: $\cshift\psi(x)=\sum_\tau \left\langle \tau, \psi(x-\tau)\right\rangle \Ket{\tau}$.

By a  simple  quantum walk $U: \bH\to\bH$ with  coin $\bbC$ we understand
\begin{equation}\label{def:quantumwalk}
U=\cshift\bbC \qquad \hbox{ with $\bbC$ a multiplication operator with unitary symbol }.
\end{equation}

Thus the corresponding one step unitary evolution $U$  of the walker is  so that
the coin matrices first reshuffle or update the coin states so that the pieces of the wave function corresponding to different internal states are then shifted to different directions, depending on the internal state : $U\psi(x)=\sum_\tau \left\langle \tau, C(x-\tau)\psi(x-\tau)\right\rangle \Ket{\tau}$.

We are working with a flux operator defined for an adapted projection, i.e.: a projection which  at each site projects onto  quantum directions:

\begin{definition}A projection valued multiplication operator $\bbP$ with symbol $P$ is called adapted if 
\[\left\lbrack P(x), P_\tau\right\rbrack=0,\quad x\in\bZ^d,\tau\in I_{2d}\]
The set of open directions at $x\in\bZ^d$ is defined as
\[I_\bbP(x):=\left\{\tau\in I_{2d}, \left\langle \tau, P(x)\tau\right\rangle\neq0\right\}\]
so $P(x)=\sum_{\tau\in I_\bbP(x)}\Ket{\tau}\Bra{\tau}$.

If $\bbP$ is an adapted projection then $\cshift^\ast\bbP\cshift$ is  also adapted  with symbol
\[\widehat{P}(x)=\sum_{\tau\in I_{2d}} P_\tau P(x+\tau)P_\tau;\]
An adapted $\bbP$ is called homogeneous in $G\subset\bZ^d$ if 
\[\dim\ran \widehat{P}(x)=\dim\ran P(x), \forall x\in G.\]
\end{definition}

\begin{example}\label{ex:homogeneous}
For $d=1$, $G=\lbrack1,\infty)\cap\bZ$ the only  two adapted projections homogeneous in $G$ and of $\dim\ran P(x)=1, x\in G$ are given by
\begin{enumerate}
\item $P_a(x):=\chi(x\in G)\Ket{1}\Bra{1}$
\item $P_b(x):=\chi(x\in G)\Ket{(-1)^x}\Bra{(-1)^x}$
\end{enumerate}
then for $x\in G$ : $\widehat{P}_a(x)=P_a(x)$, $\widehat{P}_b(x)=P_b^\perp(x)$
\end{example}

\begin{prop}\label{prop:phiqw}
Let $U$ be a quantum walk with coin $\bbC$, $\bbP$ an adapted projection, $\Phi=U^\ast \bbP U-\bbP$. It holds
\begin{enumerate}
\item the eigenvalue $1$ of $\Phi^2$ is isolated and finitely degenerate iff there exists $R>0$ such that  $\bbP$ is homogeneous in $\{x\in\bZ^d, \vert x\vert >R\}$ and it holds  in ${\bC^{2d}}$ :
\[\sup_{\vert x\vert>R}\Vert C^\ast(x)\widehat{P}(x)C(x)-P(x)\Vert\le c<1.\]
 In this case
\begin{equation}\label{eq:index}\ind(\Phi)=\sum_{\vert x\vert \le R}\left(\dim\ran \widehat{P}(x)-\dim\ran P(x)\right);\end{equation}
\item $\lbrack U,\bbP\rbrack$ is compact iff $\lim_{\vert x\vert \to\infty}\Vert C^\ast(x)\widehat{P}(x)C(x)-P(x)\Vert=0$;
\item $\sum_{x\in\bZ^d}\Vert C^\ast(x)\widehat{P}(x)C(x)-P(x)\Vert<\infty$ implies that $\lbrack U,\bbP\rbrack$ is trace class.
\end{enumerate}
\end{prop}

\begin{proof}
$1.$ $\Phi$ is an operator valued multiplication operator with self-adjoint symbol
\[\Phi(x)=C^\ast(x)\widehat{P}(x)C(x)-P(x),\]
for the spectrum it holds $\sigma(\Phi)=\bigcup_{x\in\bZ^2}\sigma(\Phi(x))$. Also the eigenspace $\ker (\Phi\pm\bI)$ is the direct sum of $\ker (\Phi(x)\pm\bI)$ which is trivial for $x$ with $\vert x\vert>R$ as $\Vert \Phi(x)\Vert<1$ there. Thus  $\dim\ker(\Phi\pm\bI)=\sum_{\vert x\vert \le R}\dim\ker (\Phi(x)\pm\bI)$ which implies the formula for the index.

$2.$ Compactness of the multiplication operator $\Phi$ is equivalent to $\Vert \Phi(x)\Vert\to 0$.

$3.$ Follows from lemma \ref{lem:traceclass}
\end{proof}

\begin{lem}\label{lem:traceclass} Let $\bbM$ be a multiplication operator. If  its symbol satisfies
\[\sum_{x\in\bZ^d}\Vert M(x)\Vert <\infty\] 
then $\bbM$ is trace class.
\end{lem}
\begin{proof}  $\bbM$ is trace class iff $\sum_x\sum_\mu \lambda_\mu(x)=\sum_x\Vert M(x)\Vert_{tr}<\infty$ for  the singular values of $M(x)$. The result follows from the equivalence of norms.

\end{proof}

As a corollary from Theorem \ref{thm:general} and Proposition \ref{prop:phiqw} we have
\begin{cor}\label{cor:perturbation}Let $\bbP$ be a projection homogeneous in $\{x\in\bZ^d, \vert x\vert >R\}$ such that
\[\sum_{x\in\bZ^d}\dim\ran \widehat{P}(x)-\dim\ran P(x)\neq0\]
let  $\bbC$ be a coin and $U=\cshift\bbC$ the associated quantum walk. Then
\begin{enumerate}
\item $C^\ast(x)\widehat{P}(x)C(x)-P(x)\to_{\vert x\vert\to\infty}0$ implies : $\sigma(U)=S^1$;
\item \label{above} $\sum_x\Vert C^\ast(x)\widehat{P}(x)C(x)-P(x)\Vert<\infty$ implies : $\sigma_{ac}(U)=S^1$;
\item for $C(x)$ as in \ref{above}. and a second coin $\bbC_2$ such that $\sum_x\Vert C(x)-C_2(x)\Vert <\infty$ it holds
\[\sigma_{ac}(\cshift\bbC_2)=S^1.\]
\end{enumerate}
\end{cor}
\begin{proof}1. and 2. follow from Theorem \ref{thm:general},  proposition \ref{prop:phiqw} and Lemma \ref{lem:traceclass}. For the third assertion remark that 
$\cshift\bbC_2-\cshift\bbC_1=\cshift(\bbC_2-\bbC_1)$ which is trace class, and the claim follows from the Birman-Krein theorem.
\end{proof}

Our basic example, see figure (\ref{fig:basicexample}) and example (\ref{ex:homogeneous}), in one dimension is
\begin{example}(Basic example) For dimension $d=1$ let $\bbC$ be a coin such that $\left\lbrack C(x),P_{(1)}\right\rbrack=0$ $\forall x\ge N\ge0$ then $\sigma_{ac}(\cshift\bbC)=S^1$.
\end{example}
\begin{proof}Let $\bbP$ be an adapted projection such that $P(x)=0$ for $x$ less or equal to $-N$ and $P(x)=P_{(1)}$ for $x$ greater or equal to $N$.Then for all directions $\tau$ : $P(x+\tau)=P(x), \forall \vert x\vert > N$ and $C^\ast(x)\widehat{P}(x)C(x)-P(x)=\Phi(x)=0, \forall \vert x\vert >N$. So $\Phi$ is trace class and 
\[trace(\Phi)=\sum_{\vert x\vert \le N}\sum_\tau\left(trace(P_\tau P(x+\tau) P_\tau)-trace(P_\tau P(x) P_\tau\right)\]
\[=trace(P_{(1)}(P(N+1)-P(-N)))+trace(P_{(-1)}(P(-N-1)-P(N)))=trace(P_{(1)})=1\]
and the claim follows from Theorem (\ref{thm:general}).
\end{proof}

\begin{figure}[hbt]
\centerline {
\includegraphics[width=8cm]{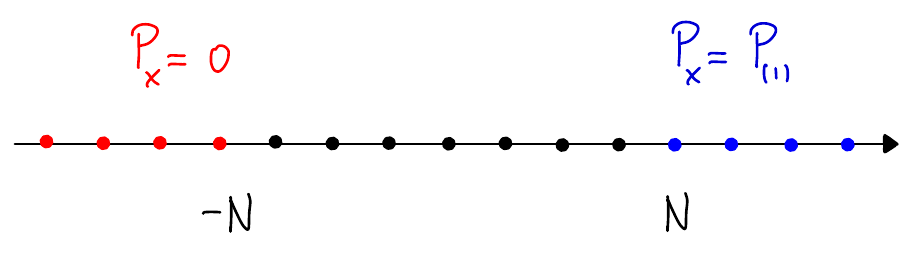}
}
\caption{Basic example}
\label{fig:basicexample}
\end{figure}

Remark that the above quantum walk $\cshift\bbC\restriction{\chi(\lbrack N,\infty))P_1}$ is a unilateral shift and the coin is totally arbitrary on the left of $N$. In particular other spectral type may coexist with the absolutely continuous spectum on $S^1$.

\subsection{Networks of leads in $\bZ^d$}

The basic example above is effortlessly transported to  a halfline in $\bZ^d$.  

\begin{example}(Halfline in $\bZ^d$) Denote $\bN_{(1)}:=\bN\times\{0\}\times\cdots\times\{0\}\subset\bZ^d$. Let $\bbC$ be a coin and $U=\cshift\bbC$ the associated quantum walk.
\begin{enumerate}
\item If  $\lbrack C(x),P_{(1)}\rbrack=0, \forall x\in\bN_{(1)}$ then it holds with the symbol $P(x):=\chi(x\in\bN_{(1)})P_{(1)}$ :  $\ind(U^\ast \bbP U-\bbP)=1$.
\item If  $\lbrack C(x),P_{(-1)}\rbrack=0, \forall x\in\bN_{(1)}$ then it holds with the symbol $P(x):=\chi(x\in\bN_{(1)})P_{(-1)}$ :  $\ind(U^\ast \bbP U-\bbP)=-1$.
\end{enumerate}
\end{example}
\begin{proof} It is sufficient to calculate $\widehat{P}$ in a 1-neighborhood of $\bN_{(1)}$ because $\Phi(x)$ vanishes outside. Here we have, see figure (\ref{fig:1neigborhood})

\begin{figure}[htb]
\centerline {
\includegraphics[width=6cm]{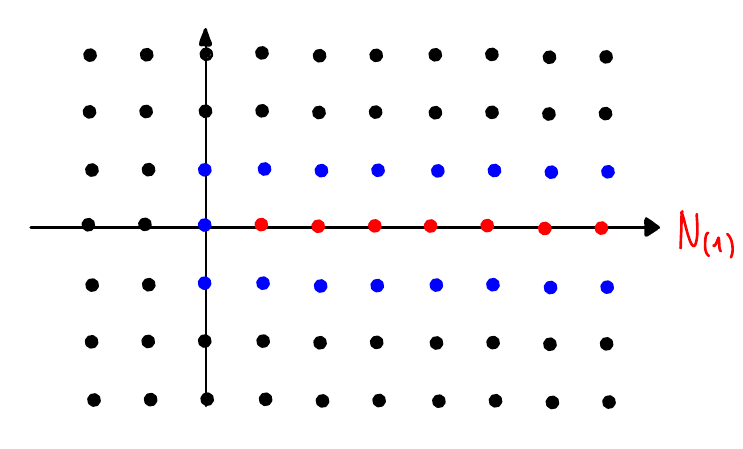}
}
\caption{1-neighborhood}
\label{fig:1neigborhood}
\end{figure}
1. 
\[P(x+(1))=
\left\{
\begin{array}{cl}
  P_{(1)}& x\in\bN_{(1)}\cup\{( 0,0,\ldots)\}     \\
  0&   \hbox{ else }
\end{array}
\right.
\]
\[P(x+(-1))=
\left\{
\begin{array}{cl}
  P_{(1)}& x\in\bN_{(1)}\setminus\{( 1,0,\ldots)\}     \\
  0&   \hbox{ else }
\end{array}
\right.
\]

\[P(x+(2))=
\left\{
\begin{array}{cl}
  P_{(1)}& x\in\bN_{(1)}+(-2)    \\
  0&   \hbox{ else }
\end{array}
\right.
\]

\[P(x+(-2))=
\left\{
\begin{array}{cl}
  P_{(1)}& x\in\bN_{(1)}+(2)    \\
  0&   \hbox{ else }
\end{array}
\right.
\]
it follows $\widehat{P}(x)=P_{(1)}\chi(x\in\bN_{(1)}\cup\{(0,0,\ldots)\})$ so  $\bbP$ is homogeneous in $\bZ^d\setminus\{(0,0,\ldots)\}$ and by formula (\ref{eq:index}) : $\ind(\Phi)=rank(P_1\chi(x=(0,0,\ldots))=1$. 

2. In the second case

$\widehat{P}(x)$ equals $P(x)$ on $\bN_{(1)}\setminus\{(1,0,\ldots)\}$ , $\bbP$ is homogeneous in $\bZ^d\setminus\{(1,0,\ldots)\}$ and \\$ind(\Phi)=-rank(P_{(-1)}\chi(x=(1,0,\ldots)))=-1$. 
\end{proof}

Generalising this example we consider outgoing and incoming leads :

\begin{definition}A path $\gamma: \bZ\ni G\to\bZ^d$ is called regular if $\vert\overrightarrow{\gamma(t-1)\gamma(t)}\vert=1, \forall t-1,t\in G$ and simple if it is injective. An unbounded regular path is called a (classical) lead. A lead is outgoing if $G=\bN=\{1,2,\ldots\}$ and a lead is incoming if $G=-\bN=\{\ldots,-2,-1\}$.

The tangent path of $\gamma$ is
\[\tau_\gamma: G\to I_{2d}, \quad \tau_\gamma(t):=\overrightarrow{\gamma(t-1)\gamma(t)}\]
with the convention : $\tau_\gamma(1):=\tau_\gamma(2)$ for an outgoing lead. An outgoing or incoming  lead is called admissible if it has no tangential selfintersections, i.e.
\[G\ni t\to (\gamma(t),\tau_\gamma(t)) \in\bZ^d\times I_{2d}\]
is injective.
\end{definition}

We use the following notations : for $x_0\in\bZ^d, \tau\in I_{2d}$ : $\Ket{x_0,\tau_0}\in\ell^2\left(\bZ^d,\bC^{2d}\right)$ is defined by $x\mapsto \delta_{x,x_0}\Ket{\tau_0}$ and 
\[\Ket{x_0,\tau_0}\Bra{x_0,\tau_0}\]
for the orthogonal projection on $span\{\Ket{x_0,\tau_0}\}$.

\begin{definition}Let $\gamma: G\to\bZ^d$ be a classical lead and $\tau_\gamma$ its tangent. The associated quantum lead is an adapted projection with symbol along $\gamma$, i.e.
\[ P_\gamma(t)=\Ket{\gamma(t),\tau_\gamma(t)}\Bra{\gamma(t),\tau_\gamma(t)},\quad t\in G.\]
\end{definition}

We construct quantum walks which implement the time evolution along a lead :

\begin{prop}\label{prop:onelead}Let  $P_\gamma$ be  an admissible outgoing (resp. incoming) quantum lead and let $\bbC$ be a coin such that
\begin{equation}\vert \left\langle \tau_\gamma(t+1),C(\gamma(t))\tau_\gamma(t)\right\rangle\vert=1\quad \forall t\in G\ s.t.\  t+1\in G.\label{cond:c}\end{equation}
Then $P_\gamma(t)P_\gamma(s)=0, \forall t\neq s \in G$ and for the quantum walk $U=\cshift\bbC$ it holds:
\[U^nP_\gamma(1){U^\ast}^n=P_\gamma(1+n) \quad \forall n\in\bN \hbox{ if } \gamma \hbox{ is outgoing }, \]
\[{U^\ast}^nP_\gamma(-1)U^n=P_\gamma(-1-n) \quad \forall n\in\bN \hbox{ if } \gamma \hbox{ is incoming }. \]
In particular $\ran P_\gamma(1)$ (resp. $\ran P_\gamma(-1)$) is a wandering subspace for $U$ (resp. $U^\ast$) if $\gamma$ is outgoing (resp. incoming), and for the total projector \[\bbP_\gamma:=\sum_{t\in G}P_\gamma(t) \hbox{ and the flux } \Phi=U^\ast \bbP_\gamma U -\bbP_\gamma,\quad G=\bN\ (\hbox{ resp. } G=-\bN)\]
it holds
\[\ind(\Phi)=
\left\{
\begin{array}{rl}
  1&   \hbox{ if } \gamma \hbox{ is outgoing }   \\
  -1&   \hbox{ if } \gamma \hbox{ is incoming } .  \\
\end{array}
\right.
\]
\end{prop}
\begin{proof}By the admissibility assumptions $\Ket{\gamma(t_0),\tau_\gamma(t_0)}\perp\Ket{\gamma(t_1),\tau_\gamma(t_1)}, \forall t_0,t_1\in G$. It follows that $\bbP_\gamma$ is well defined as a strong limit of orthogonal projections.

For any $t\in {G}$ ( except $t=-1$ in the incoming case) it holds for a phase $\varphi_t\in S^1$:
\[U\Ket{\gamma(t),\tau_\gamma(t)}=\varphi_t \cshift\Ket{\gamma(t),\tau_\gamma(t+1)}=\varphi_t\Ket{\gamma(t+1),\tau_\gamma(t+1)}.\]
This proves the wandering subspace property and furthermore :
\[U \bbP_\gamma U^\ast=
\left\{
\begin{array}{ll}
\bbP_\gamma -\Ket{\gamma(1),\tau_\gamma(1)}\Bra{\gamma(1),\tau_\gamma(1)}   &  \hbox{ if }\gamma\hbox{ is outgoing} \\
\bbP_\gamma + U \Ket{\gamma(-1),\tau_\gamma(-1)}\Bra{\gamma(-1),\tau_\gamma(-1)} U^\ast  &  \hbox{ if }\gamma\hbox{ is outgoing.}    
\end{array}
\right.
\]
It follows for $\Phi=U^\ast \bbP_\gamma U-\bbP_\gamma$:

\begin{equation}\label{eq:phiinout}
\Phi=
\left\{
\begin{array}{ll}
U^\ast\Ket{\gamma(1),\tau_\gamma(1)}\Bra{\gamma(1),\tau_\gamma(1)}U   &  \hbox{ if }\gamma\hbox{ is outgoing} \\
-\Ket{\gamma(-1),\tau_\gamma(-1)}\Bra{\gamma(-1),\tau_\gamma(-1)}   &  \hbox{ if }\gamma\hbox{ is outgoing.}    
\end{array}
\right.
\end{equation}
and 
\[\ind(\Phi)=
\left\{
\begin{array}{rl}
  1&   \hbox{ if } \gamma \hbox{ is outgoing }   \\
  -1&   \hbox{ if } \gamma \hbox{ is incoming } .  \\
\end{array}
\right.\]
\end{proof}

Because of the locality or the construction we can consider an arbitrary finite number of leads and design currents.

\begin{figure}[h]
\centerline {
\includegraphics[width=8cm]{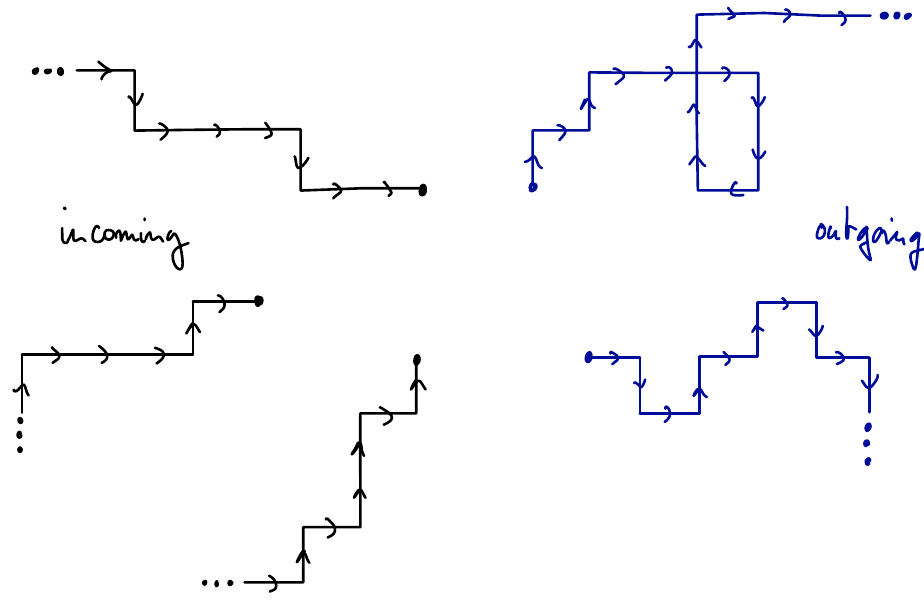}
}
\caption{Networks of leads}
\end{figure}

\begin{definition}We say that leads do not cross tangentially if for any $(x,\tau)\in\bZ^d\times I_{2d}$ there is at most one lead $\gamma$ such that $(x,\tau)\in\ran(\gamma,\tau_\gamma)$
\end{definition}

\begin{prop}\label{prop:severalleads} For $n_i,n_o\in\bN$ consider $n_i$ incoming leads $\gamma_j$ and $n_o$ outgoing leads $\rho_k$ which do not cross tangentially. Let $\bbC$ be a coin satisfying condition (\ref{cond:c}) along all leads and $U=\cshift\bbC$ the associated walk. Then the total projections on the quantum leads are mutually orthogonal and for 
\[\bbP=\sum_j P_{\gamma_j}+\sum_k P_{\rho_k} \hbox{ and the flux } \Phi=U^\ast \bbP U-\bbP\]
it holds
\[\ind(\Phi)=n_o-n_i.\]
\end{prop}
\begin{proof}At a given $x\in\bZ^d$ no more than $2d$ paths can cross non tangentially so there  exists a $C(x)\in U(2d)$ satisfying  (\ref{cond:c}). The non-tangential condition implies $P_\alpha(t)P_\beta(s)=0, \forall t\neq s, \forall \alpha=\beta\in\{\gamma_j,\rho_k\}$. The result then follows because  equation (\ref{eq:phiinout}) implies: $\ind(\Phi_{P_1+P_2})=\ind(\Phi_1)+\ind(\Phi_2)$
\end{proof}

\begin{rem}Note that the coin $C(x)$ is arbitrary for $x$ outside the leads and that $n_o-n_i\neq0$ implies $\sigma_{ac}(U)=S^1$; also the spectral result remains true if condition (\ref{cond:c}) is only satisfied asymptotically in the sense of Corollary \ref{cor:perturbation}, point 3.
\end{rem}
\begin{examples}In $\bZ^2$ consider


\begin{enumerate}
\item two outgoing leads: $\rho_1(n)=(0,n-1), \forall n\in\bN$, $\rho_2(n)=(-n+1,0), \forall n\in\bN$ with tangents $\tau_{\rho_1}(n)=(0,1)$, $\tau_{\rho_2}(n)=(-1,0)$. Then for $\Phi=U^\ast(P_{\rho_1}+P_{\rho_2})U-(P_{\rho_1}+P_{\rho_2})$ : $\ind(\Phi)=2$; 
\item an incoming and an outgoing lead: $\gamma_1(n)=(0,-n-1), \forall n\in-\bN$, $\rho_2(n)=(-n+1,0), \forall n\in\bN$ with tangents $\tau_{\gamma_1}(n)=(0,-1)$, $\tau_{\rho_2}(n)=(-1,0)$. Then for $\Phi=U^\ast(P_{\gamma_1}+P_{\rho_2})U-(P_{\gamma_1}+P_{\rho_2})$ : $\ind(\Phi)=1-1=0$.
\end{enumerate}
\end{examples}

\subsection{Networks of leads tangential to bulk boundaries}
We now design quantum walks which propagate along a network of leads on a surface bounding a half-space and such that there is no flow out of the halfspace. We start by constructing a projection and a coin such that there is no flow out the halfspace and add conducting leads on the surface in a second step.

Consider the lattice 
\[\bZ^{d+1} \hbox{ split in two halfs } \bZ^{d+1}_\pm:=\bZ^d\times (\pm \bN) \hbox{ separated by the boundary } \Gamma=\bZ^d\times\{0\}. \]
The tangent space is 
\[T\Gamma=\{(\pm j), j\in\{1,\ldots,d\}\}\subset I_{2(d+1)}, \hbox{ and the in-(out-)ward normal } \pm N=(\pm (d+1));\]
we use the same notations for the quantum directions $T\Gamma=\{\Ket{\pm j}\}_{j\in\{1,\ldots,d\}}\subset \bC^{2(d+1)},  \Ket{\pm N}=\Ket{\pm (d+1)}$

\begin{figure}[htb]
\centerline{
\includegraphics[width=8cm]{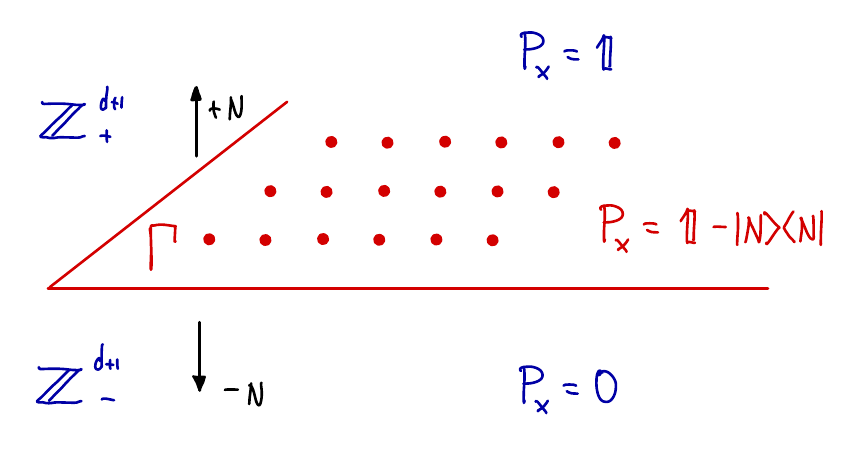}
}
\caption{Splitting by a hypersurface}
\end{figure}

\begin{lem} Let $\bbP$ be the adapted projection with symbol
\[P(x)=
\left\{
\begin{array}{lr}
  \bI &   x\in\bZ_+^{d+1} \\
  0&   x\in\bZ_-^{d+1}  \\
  \Ket{-N}\Bra{-N}+P_{T\Gamma}=\bI-\Ket{N}\Bra{N}&   x\in\Gamma  
\end{array}
\right.
\]
(with $P_{T\Gamma}$ the projection on the tangential directions) and $\bbC$ any coin whose symbol leaves the tangent space invariant and which reflects the normal on $\Gamma$, i.e.:
\[C(x) T\Gamma=T\Gamma,  C(x)\Ket{\pm N}=\Ket{\mp N}\quad\forall x\in\Gamma\]
then $\bbP$ is homogeneous in $\bZ^{d+1}$ and for $U=\cshift\bbC$ it holds
\[\Phi=U^\ast \bbP U-\bbP=0.\]
\end{lem}

Note that because the coin is reflecting on the surface a Quantum Walker can move tangentially to, but cannot cross $\Gamma$.

\begin{proof}
By homogeneity it is sufficient to consider the $1$-neighborhood of $\Gamma=\Gamma\cup (\Gamma+N) \cup (\Gamma -N)$.
For $x\in\Gamma$ we have 
\[\widehat{P}(x)=\sum_{\tau\in\Gamma}P_\tau\bI P_\tau+P_N\underbrace{P(x+N)}_{=\bI}P_N+P_{-N}\underbrace{P(x-N)}_0 P_{-N}=\bI\restriction{T\Gamma}+P_N,\]
for $x\in\Gamma+N$ : 
\[\widehat{P}(x)=\sum_{\tau\in I_{2(d+1)}\setminus\{-N\}}P_\tau+\underbrace{P_{-N}P(x-N)P_{-N}}_{P_{-N}}=P(x)\]
and for $x\in\Gamma-N$ :
\[\widehat{P}(x)=P_NP(x+N)P_N=0=P(x).\]
It follows that $rank\widehat{P}(x)=rank P(x), \forall x$ and from the properties of $C$:
\[\Phi(x)=C^\ast(x)\widehat{P}(x)C(x)-P(x)=0, \forall x.\]
\end{proof}

We now are able to consider a network of leads on the surface $\Gamma$ and a quantum walk such that there is no crossing from $\bZ_+^{d+1}$ to $\bZ_-^{d+1}$ but which propagates on $\Gamma$ with arbitrary index of $\Phi$:

\begin{figure}[htb]
\centerline{
\includegraphics[width=8cm]{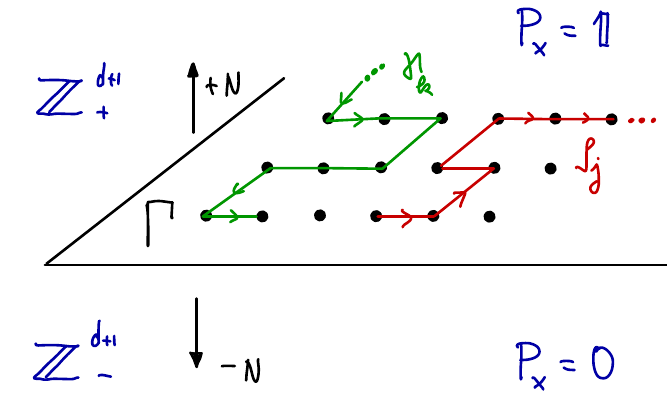}
}
\caption{Leads on a surface}
\end{figure}

\begin{thm}\label{thm:quantumleads}
For $n_i,n_o\in\bN$ consider $n_i$-incoming leads $\gamma_j$ and $n_o$ outgoing leads $\rho_k$ on $\Gamma$ which are admissible and do not cross tangentially. 

Let $\bbP$ be the adapted projection with symbol
\[P(x)=
\left\{
\begin{array}{lr}
  \bI &   x\in\bZ_+^{d+1} \\
  0&   x\in\bZ_-^{d+1}  \\
  \Ket{-N}\Bra{-N}&   x\in\Gamma  
\end{array}
\right.
\]

and, as in proposition \ref{prop:severalleads}
 \[\bbP_L=\sum_j P_{\gamma_j}+\sum_k P_{\rho_k}\]
 the total projection on the leads. Let $\bbC$ be any coin which is reflecting on $\Gamma$ :
 \[C(x)\Ket{\pm N}=\Ket{\mp N}\qquad \forall x\in\Gamma\]
 and such that condition (\ref{cond:c}) is verified for all leads. Then for
 \[Q=\bbP+\bbP_L, \quad \Phi=U^\ast QU-Q\] \hbox{ it holds }\[ \ind(\Phi)=n_0-n_i.\]
 Furthermore
\begin{enumerate}
\item $U\ell^2(\bZ_\pm^{d+1}; \bC^{2d})\subset \ell^2(\bZ_\pm^{d+1}\times\Gamma; \bC^{2d})$
\item  $\ran P_{\rho_k}(1)$ is a wandering subspace for $U$
\item $\ran P_{\gamma_j}(-1)$ is a wandering subspace for $U^\ast$
\end{enumerate}
\end{thm}
\begin{proof}$\bbP_L$ contains only tangential directions while $P$ contains none,  $\ind(\Phi_{\bbP_L})+\ind(\Phi_{\bbP})=\ind(\Phi_Q)$. For the same reasons the two conditions on the coin are compatible and the dynamical results follow from the previous results.
\end{proof}

\begin{rem}
\begin{enumerate}
\item Note that the coin $C(x)$ is arbitrary for $x$ outside the leads and that $n_o-n_i\neq0$ implies $\sigma_{ac}(U)=S^1$; also the spectral result remains true if condition (\ref{cond:c}) is only satisfied asymptotically in the sense of Corollary \ref{cor:perturbation}, point 3.
\item It is straightforward the extend the above results to hypersurfaces more general than $\Gamma$.
\end{enumerate}

\end{rem}

\section{Transport perpendicular to bulk boundaries, the Chalker--Coddington model}\label{section:cc}

This section is devoted to spectral and transport results for Chalker-Coddington models defined in the  two-dimensional plane. It is an extension of the spectral analysis for the model  in a strip via a flux study lead in \cite{ABJ4}. For background on the model we refer to  \cite{ABJ4} and references therein. Let 
\[\left\{S_{j,2k}\right\}_{j,k\in\bZ}\]
be a collection of unitary $2 \times 2$ matrices $S_{j,2k}\in U(2)$ called scattering matrices, even or odd according to the parity of $j$; they define the unitary operator which characterises the
 Chalker-Coddington model 
\[U:\ell^2(\bZ^2;\bC)\to\ell^2(\bZ^2;\bC)\]
in the following way :
denoting by  $\ket{j,k}$ the canonical basis vectors of $l^2(\bZ^2;\bC)$,  $U$  is defined according to figure (\ref{fig:scatteringnetwork}) by:

\begin{align}
&\begin{pmatrix} U |2j,2k\rangle \cr U |2j+1,2k-1\rangle \end{pmatrix} := S_{2j,2k} \begin{pmatrix} |2j,2k-1\rangle \cr |2j+1,2k\rangle \end{pmatrix},\notag \\
&\begin{pmatrix} U |2j+1,2k\rangle \cr U  |2j+2,2k+1\rangle \end{pmatrix} := S_{2j+1,2k} \begin{pmatrix} |2j+2,2k\rangle \cr |2j+1,2k+1\rangle \end{pmatrix}. \label{def:UCC}
\end{align}
\begin{figure}[hbt]
\centerline {
\includegraphics[width=8cm]{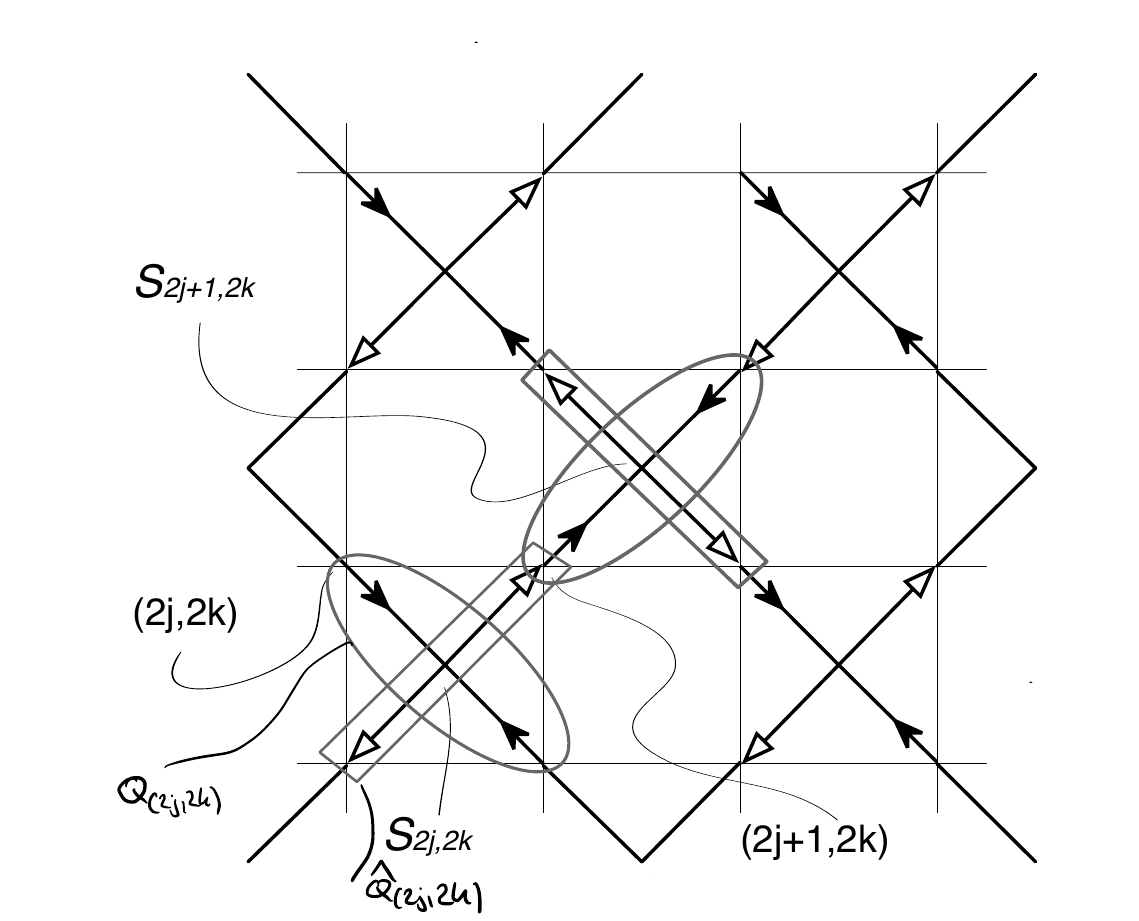}
}
\caption{A Chalker--Coddington model with its incoming (solid arrows) and outgoing links.}
\label{fig:scatteringnetwork}
\end{figure}

The in- and outgoing subspaces of the scattering matrices play a prominent role in the following: 

for $z=(j,2k)\in\bZ\times2\bZ$ denote $Q_z$ the orthogonal projections on
\begin{eqnarray*}
&&\ran Q_{2j,2k}=span\left\{\Ket{2j,2k},\Ket{2j+1,2k-1}\right\},\\
&&\ran Q_{2j+1,2k}=span\left\{\Ket{2j+2,2k+1},\Ket{2j+1,2k}\right\}.
\end{eqnarray*}
and their images $\widehat{Q}_z=U Q_zU^\ast$ on
\begin{eqnarray*}
&&\ran \widehat{Q}_{2j,2k}=span\left\{\Ket{2j+1,2k},\Ket{2j,2k-1}\right\},\\
&&\ran \widehat{Q}_{2j+1,2k}=span\left\{\Ket{2j+1,2k+1},\Ket{2j+2,2k}\right\}.
\end{eqnarray*}

\begin{rems}\begin{enumerate}
\item In the original paper \cite{cc} the moduli of the elements of $S_{j,2k}$ were constant and the phases random variables. Here we have no restrictions.
\item Note that $Q_z,\widehat{Q}_z$ do not depend on the collection of scattering matrices.
\end{enumerate}
\end{rems}

For the flux operator we observe the following :

\begin{prop}\label{prop:indexformula} Let $G$ be a subset of $\bZ^2$ and $P$ the multiplication operator $P=\chi(x\in G)$, 
\[\Phi=U^\ast P U -P.\]
Then it holds

\begin{enumerate}
\item $\left\lbrack\Phi,Q_z\right\rbrack=0 \quad \forall z\in\bZ\times2\bZ$;
\item $\Phi Q_z=U^\ast (P\widehat{Q}_z)U- PQ_z$;
\item $\ind(\Phi)$ is well defined iff for a $c,R>0$ : $\sup_{\vert z\vert> R}\Vert \Phi Q_z\Vert\le c<1$,  then it holds
\[\ind(\Phi)=\sum_{z, \vert z\vert \le R} \dim\ran(P\widehat{Q}_z)-\dim\ran(PQ_z).\]
\end{enumerate}
\end{prop}

\begin{proof} By definition $Q_z,\widehat{Q}_z$ are multiplication operators on $\ell^2(\bZ^2;\bC)$ thus commute with $P$, or
\[\lbrack\Phi,Q_z\rbrack=\lbrack U^\ast P U,Q_z\rbrack=U^\ast\lbrack P, \widehat{Q}_z\rbrack U=0.\]
Also $\ker\left(\Phi^2-\bI\right)=\bigoplus_{z\in\bZ\times2\bZ}\ker\left((\Phi^2-\bI)Q_z\right)$ and $\sigma\left(\Phi^2\right)=\bigcup_{z\in\bZ\times2\bZ}\sigma\left(\Phi^2Q_z\right)$. By  assumption $dist\left(\{1\},\bigcup_{\vert z\vert>R}\sigma(\Phi^2Q_z)\right)>0$ thus $1$ is an isolated finite dimensional  eigenvalue and 
\[\ind(\Phi)=\sum_{\vert z\vert\le R} \dim\ker\left((\Phi-\bI)Q_z\right)-\sum_{\vert z\vert\le R}\dim\ker\left((\Phi+\bI)Q_z\right)=\]\
\[=\sum_{\vert z\vert\le R}\ind(\Phi Q_z)=\sum_{\vert z\vert\le R}trace(\Phi Q_z)=\sum_{z, \vert z\vert \le R} \dim\ran(P\widehat{Q}_z)-\dim\ran(PQ_z).\]
\end{proof}
To define an appropriate projection $P$ we use some graph terminology.  To the given $U$ be can associate the directed graph
\[G=(V,E)\]
with vertices in $V=\bZ^2$ and whose directed edges are defined by
\[\overrightarrow{xy}\in E \hbox{ \it iff } \langle y,U_c\ x\rangle\neq0\]
where $U_c$ is the critical model which has all its scattering matrices equal to $\frac{1}{\sqrt{2}}\left(
\begin{array}{cc}
  1   &-1   \\
  1   &  1  
\end{array}
\right)$.

Some features of the given $U$ are then described by attributing 
\[\hbox{ the weight }\vert\langle y,U x\rangle\vert \hbox{ to the edge } \overrightarrow{xy}\in E\]
and we call an edge open of its weight does not vanish, see figure (\ref{fig:weights}).

\begin{figure}[hbt]
\centerline {
\includegraphics[width=6cm]{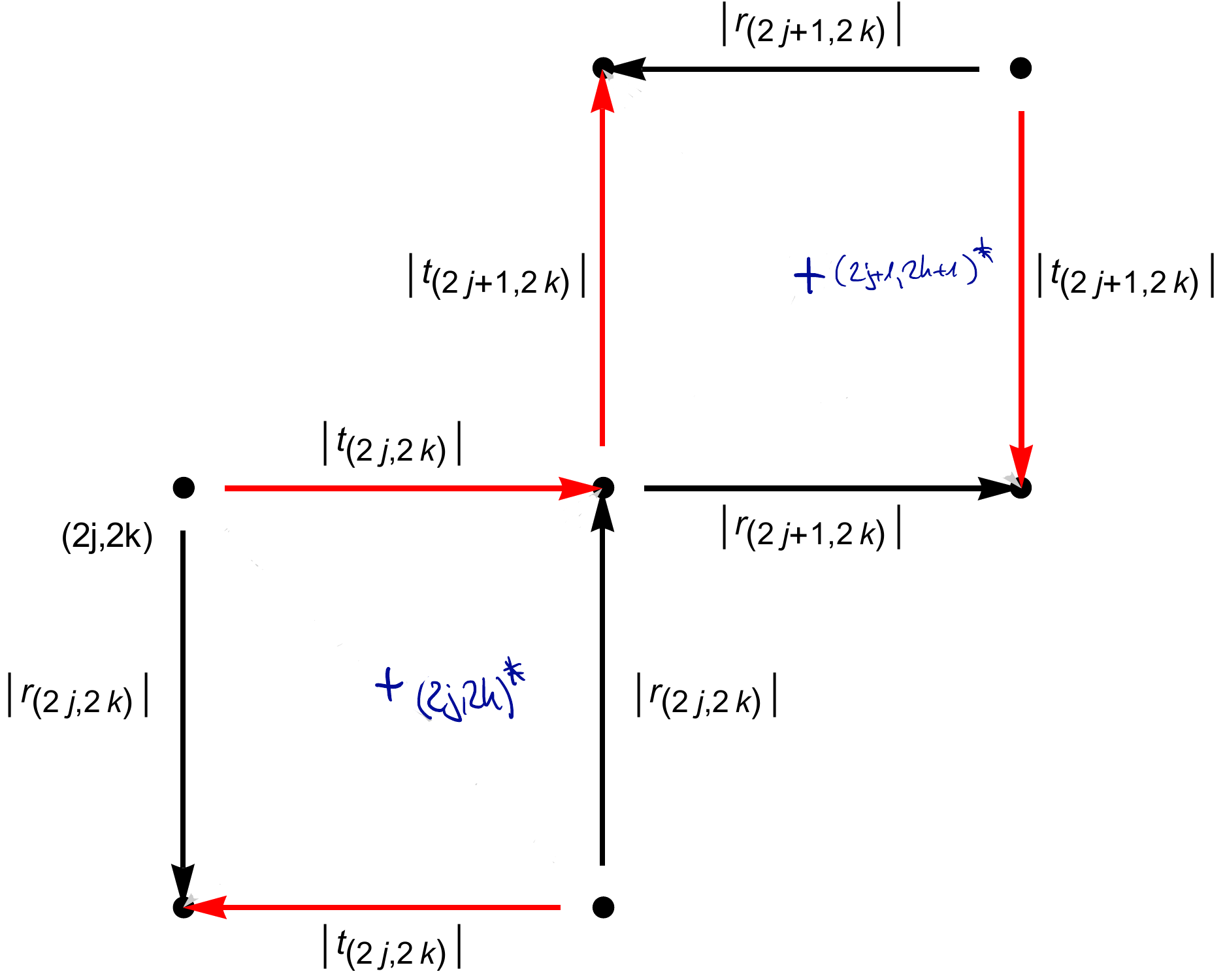}
}
\caption{The weights of $G$.}
\label{fig:weights}
\end{figure}

The dual graph 
\[G^\ast=(V^\ast,E^\ast)\]
has its vertices at the center of the faces of $G$
\[V^\ast=\bZ^2+\frac{1}{2}(1,-1)\]
and its edges (which we call links to distinguish) connect all neighbours at distance $1$ thus bisecting the edges of $G$.
We use the notations: 
\[V^\ast\ni v^\ast=(x_1^\ast,x_2^\ast):=(x_1,x_2)+\frac{1}{2}(1,-1)\] for integers $x_1,x_2$ and say that
\[x_j^\ast\in Even \hbox{ if } x_j\in2\bZ, \qquad x_j^\ast\in Odd \hbox{ if } x_j\in2\bZ+1. \]  Our  parameters for the scattering matrices $S_z$, $z\in\bZ\times 2\bZ$ are:
\[S_z=q_z\left(
\begin{array}{cc}
  r_z   &-t_z   \\
  \bar{t_z}   &  \bar{r_z}  
\end{array}
\right)\quad \hbox{ with } q_z\in S^1, r_z,t_z\in\bC \hbox{ s.t. } \vert r_z\vert^2+\vert t_z\vert^2=1.\]

\begin{figure}[hbt]
\centerline {
\includegraphics[width=8cm]{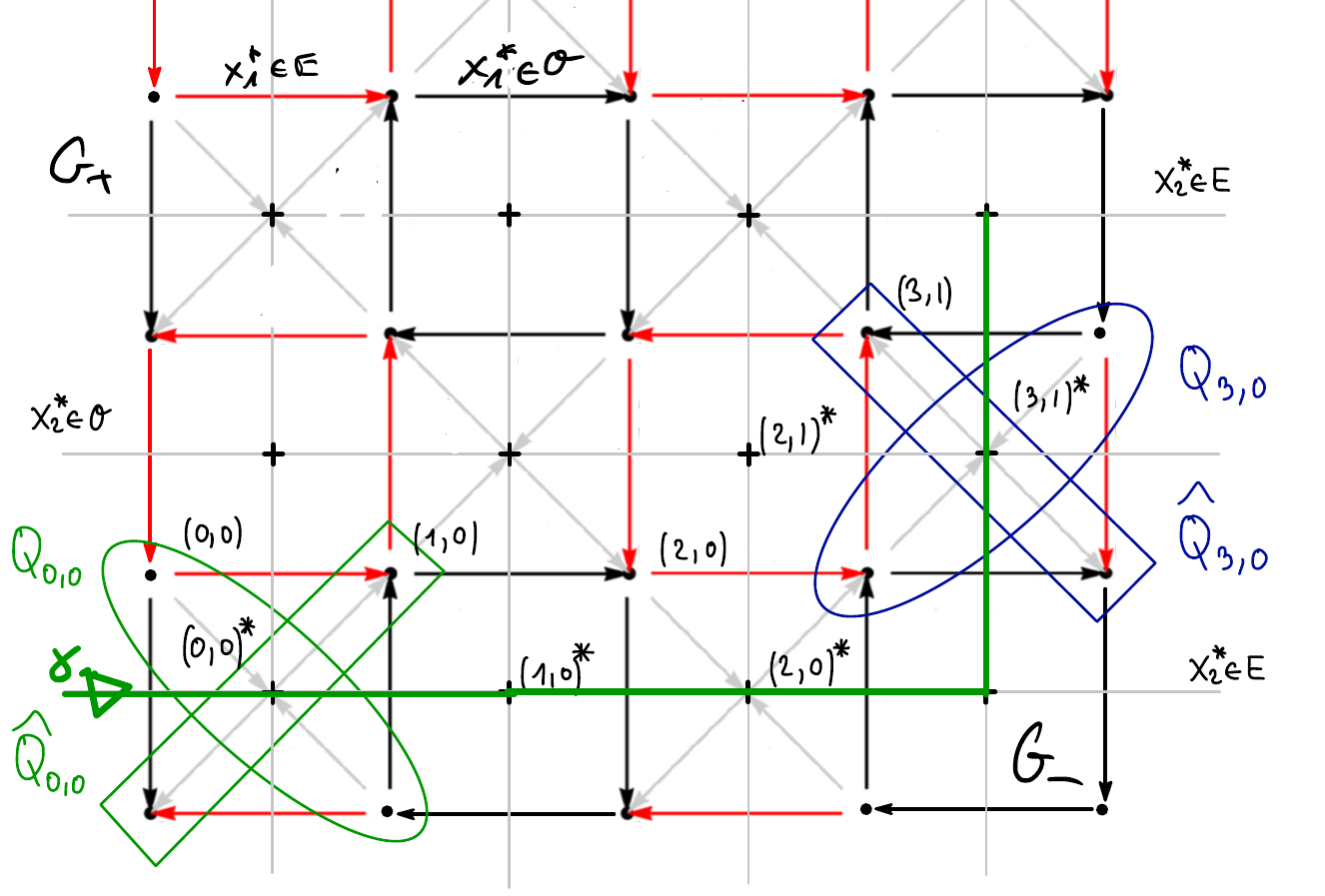}
}
\caption{Edges of weight $\vert r\vert$  (black) are crossed by links on $\gamma$ incident to $(x_1^\ast, x_2^\ast)$ which are : horizontal if  $x_2^\ast\in Even$ or 
vertical if  $x_1^\ast\in Odd$. Red edges are of weight $\vert t\vert$.}
\label{fig:notations}
\end{figure}
Now consider a partition of $G$ in two infinite connected subgraphs $G_+$ and $G_-$ and consider a path in $G^\ast$ which bisects the edge boundary of $G_+$.

 To be specific apprehend such a path as a continuously parametrized injective piecewise unit speed curve of segments of length one 
\[\gamma: \bR\to G^\ast\subset\bR^2 \hbox{ such that } \gamma(t)\in V^\ast \hbox{ and } \overrightarrow{\gamma(t)\gamma(t+1)}\in E^\ast \hbox{ for } t\in\bZ\]
oriented such that $G_+$ is to the left of $\gamma$; we call it an admissible\ path.
Let $V_+$ be the set of vertices of $G_+$  and consider its projection and flux operator in $\ell^2(\bZ^2; \bC)$ :
\[P_\gamma=\chi(x\in V_+), \quad \Phi_\gamma=U^\ast P_\gamma U-P_\gamma.\]

In the following we shall give sufficient conditions for non-trivial $\ind(\Phi_\gamma)$. By the general theory of Theorem \ref{thm:general} and Proposition \ref{lem:wandering} we know that all vertices which are not incident to the edge belong to  $\ker(\Phi_\gamma)$; in view of proposition \ref{prop:indexformula} we  label the set of interesting vertices as follows:

\begin{definition}Let $\gamma:\bR \to G^\ast$ be an admissible path. We call $E_\gamma\subset E$
the set of edges bisected by $\gamma$ and $V_\gamma$ the labels of subspaces $\ran Q_z$ which contain  vertices incident to $E_\gamma$ :
\[V_\gamma:=\left\{z\in\bZ\times2\bZ; P_\gamma^\perp U_c P_\gamma Q_z\neq0 \hbox{ or } P_\gamma U_c P_\gamma^\perp Q_z\neq0\right\}.\]
$\gamma$ is called  an $r$-path in $I\subset\bR$ if its restriction to $I$ bisects only edges of weight $\vert r\vert$. A $t$-path in $I$ is defined analogously.
\end{definition}
\begin{rem}\label{rem:directions}Observe that according to figure(\ref{fig:notations}) any link of an $r$-path can be incident to  dual vertices in $Odd\times Even$ but it has to be horizontal at  $Even \times Even$ and vertical at $Odd\times Odd$. For a $t$-path the opposite is true.
\end{rem}
\begin{lem}\label{lem:estimate}Let $\gamma$ be an admissible  $r$-path in $\bR$; then it holds for the operator and the trace norm: 
\[\Vert\Phi_\gamma\Vert\le\sup_{z\in V_\gamma}\vert r_z\vert,\quad \Vert\Phi_\gamma\Vert_1\le const. \sum_{z\in V_\gamma}\vert r_z\vert.\]
\end{lem}

\begin{proof}  We suppress the subscript $\gamma$ to $\Phi$ and $P$ for the proof. Only states with label in $V_\gamma$ contribute and   $\lbrack \Phi^2,Q_z\rbrack=0$ so we have
\[\Phi^2=PU^\ast P^\perp U P+P^\perp U^\ast P U P^\perp=\sum_{z\in \bZ\times2\bZ} \Phi^2 Q_z=\sum_{z\in V_\gamma} \Phi^2 Q_z=\sum_{z\in V_\gamma} \vert r_z\vert^2 Q_z.\]
Now $\lbrack\Phi^2,P\rbrack=0$ and $\gamma$ crosses $Q_z$ horizontally or vertically and only  edges of weight $\vert r_z\vert$  thus $\Phi^2 Q_z=\vert r_z\vert^2 Q_z$, see figure (\ref{fig:phi2}). It follows that  $\vert \Phi_\gamma\vert=\sum_{z\in V_\gamma}\vert r_z\vert Q_z$ which implies the assertion.

\begin{figure}[hbt]
\centerline {
\includegraphics[width=6cm]{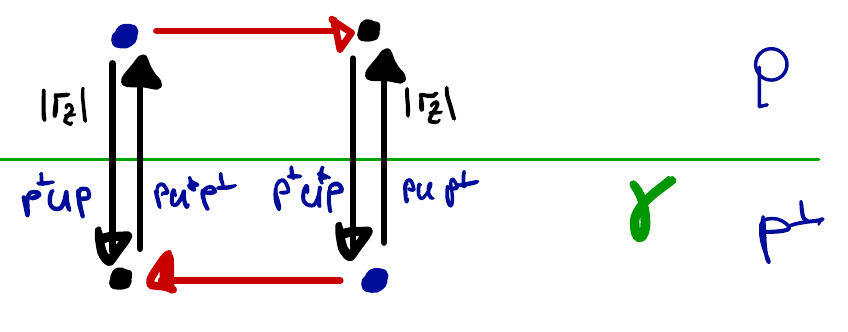}
}
\caption{The action of $\Phi^2Q_z$ for an $r$-path}
\label{fig:phi2}
\end{figure}

\end{proof}

\begin{prop}\label{prop:rpath} Let $\gamma$ be an admissible $r$-path in $I=\bR\setminus\lbrack-N,N\rbrack$ for an $N>1$ and such that  for a $c>0$ 
\[\vert r_z\vert\le c<1 \quad \forall z\in V_{\gamma\restriction I}\]
then $\ind(\Phi_\gamma)$ is well defined and it holds
\[\ind(\Phi_\gamma)=0.\]

The analogous statement holds true for a $t$-path.
\end{prop}

\begin{proof}  We suppress the subscript $\gamma$ to $\Phi$ and $P$ for the proof. It follows from Proposition (\ref{prop:indexformula}) and Lemma (\ref{lem:estimate}) that $\ind(\Phi)$ is well defined because 
\[\sum_{z\in V_\gamma} \Phi Q_z-\sum_{z\in V_{\gamma\restriction I}} \Phi Q_z\] is of finite rank 
and
\[\ind(\Phi)=\sum_{z\in V_{\gamma\restriction_{[-N,N]}}} \dim\ran(P\widehat{Q}_z)-\dim\ran(PQ_z).\]

 If $\gamma$ is already an $r$-path in all $\bR$  then $ \dim\ran(P\widehat{Q}_z)=\dim\ran(PQ_z)=1, \forall z\in V_\gamma$ because only $r$ edges are bisected thus $P$ splits $\ran Q_z$ and $\ran\widehat{Q}_z$ horizontally or vertically according parity, c.f. (\ref{rem:directions}). Thus  $\ind(\Phi_\gamma)=0$.
 
 If not,  choose two integers $N_\pm\in\bZ$ such that $ \pm N_\pm>\pm N$ and such that $\gamma(N_\pm)\in Odd \times Even$ thus all links incident to the dual vertices $\gamma(N_\pm)$  bisect edges of weight $\vert r\vert$. Then we can define a new admissible path $\widehat\gamma$  replacing the part $\gamma\restriction_{(-N_-,N_+)}$ by an $r$ path inside $G_+$ connecting $\gamma(N_-)$ to $\gamma(N_+)$. The difference $P_\gamma-P_{\widehat\gamma}$ is of finite rank thus 
 \[\ind(\Phi_\gamma)=\ind(\Phi_{\widehat\gamma})=0.\]
\end{proof}
\begin{thm} \label{thm:main} Let $U$ be a Chalker Coddington model such that there exists  an admissible path $\gamma$  which is an $r$-path in $(-\infty, -N\rbrack$ and a  $t$-path in $\lbrack N, \infty)$ for an integer $N>1$, and such that  for a $c>0$ 
\[\vert r_z\vert\le c<1 \quad \forall z\in V_{\gamma\restriction (-\infty, -N\rbrack} \hbox{ and } \vert t_z\vert\le c<1 \quad \forall z\in V_{\gamma\restriction \lbrack N, \infty)}.\]
Then $\ind(\Phi_\gamma)$ is well defined and it holds
\[\vert\ind(\Phi_\gamma)\vert=1.\]
In addition
\begin{enumerate}
\item $\Phi_\gamma$ is compact iff \quad$\lim_{s\to-\infty}{r_{\gamma(s)}}=0=\lim_{s\to\infty}{t_{\gamma(s)}}$
and then  $\sigma(U)=S^1$
\item $\Phi_\gamma$ is trace class iff  \quad$\sum_{z\in V_{\gamma\restriction (-\infty, -N\rbrack} }\vert r_z\vert + \sum_{z\in V_{\gamma\restriction \lbrack N, \infty)}} \vert t_z\vert <\infty$
and then 
\[\sigma_{ac}(U)=S^1\]
and a trace class pertubation of $U$ contains a shift operator.
\end{enumerate}

\end{thm}
\begin{proof}If $\gamma$ is a path which switches from $r$ to $t$ on one site, i.e : if $\gamma$ is an $r$-path in $(-\infty, 0)$ and a $t$-path in $[0,\infty)$, then as a corollary of Propositions \ref{prop:indexformula} and \ref{prop:rpath} we have
\[\ind(\Phi)=\sum_{z\in V_{\gamma\restriction_{[-N,N]}}} \dim\ran(P\widehat{Q}_z)-\dim\ran(PQ_z)=\dim\ran(P\widehat{Q}_{z_0})-\dim\ran(PQ_{z_0})\]
where $z_0$ is such that  $\gamma(0)$ is the center of the face defined by the vertices in $\ran\widehat{Q}_{z_0}$ and $\ran Q_{z_0}$. By the symmetry of the problem these two edges are either incoming or outgoing to $G_+$ thus  $\dim\ran PQ_{z_0}=1$ and $\dim\ran P\widehat{Q}_{z_0}=0$ or the other way round, see figure (\ref{fig:crossover}). Thus $\vert \ind(\Phi_\gamma)\vert=1$.

If not,  choose two integers $N_\pm\in\bZ$ such that $ \pm N_\pm>\pm N$ and such that $\gamma(N_-)\in Odd \times Even$ and $\gamma(N_+)\in Even \times Odd$ thus all links incident to the dual vertices $\gamma(N_-)$  bisect edges of weight $\vert r\vert$ and all links incident to the dual vertices $\gamma(N_+)$  bisect edges of weight $\vert t\vert$. Then we can define a new admissible path $\widehat\gamma$  replacing the part $\gamma\restriction_{(-N_-,N_+)}$ by an $r$ path inside $G_+$ connecting $\gamma(N_-)$ to $\gamma(N_+)$, c.f. figure (\ref{fig:joinrt}). The difference $P_\gamma-P_{\widehat\gamma}$ is of finite rank thus 
 \[\ind(\Phi_\gamma)=\ind(\Phi_{\widehat\gamma}).\]
 The additional assertions follow from lemma \ref{lem:estimate} and theorem \ref{thm:general}.
\end{proof}

As by unitarity it holds $\vert r_z\vert^2+\vert t_z\vert^2=1$ it follows:
\begin{cor} The critical Chalker Coddington model defined by $\vert r_z\vert=\vert t_z\vert=\frac{1}{\sqrt{2}}, \forall z$ admits a projection $P$ with non-trivial index
\[\ind(\Phi)\neq0\]
for any distribution of phases of its scattering matrices. 

More generally: a Chalker Coddington model defined by a collection of scattering matrices such that  for  $c_1,c_2\in(0,1)$ and all $z$:  $0<c_1<\vert r_z\vert<c_2<1$, admits a projection with non trivial index.
\end{cor}

\begin{figure}
\centerline {
\includegraphics[width=5cm]{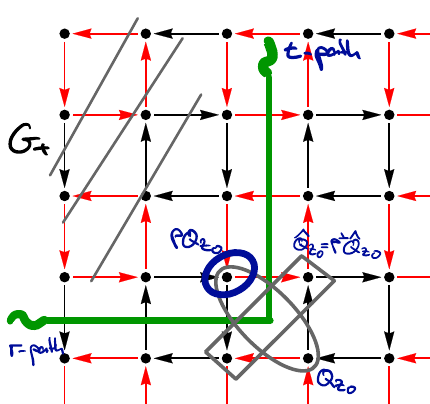}
}
\caption{Crossover from $r$ to $t$.}
\label{fig:crossover}
\end{figure}

\begin{figure}
\centerline {
\includegraphics[width=8cm]{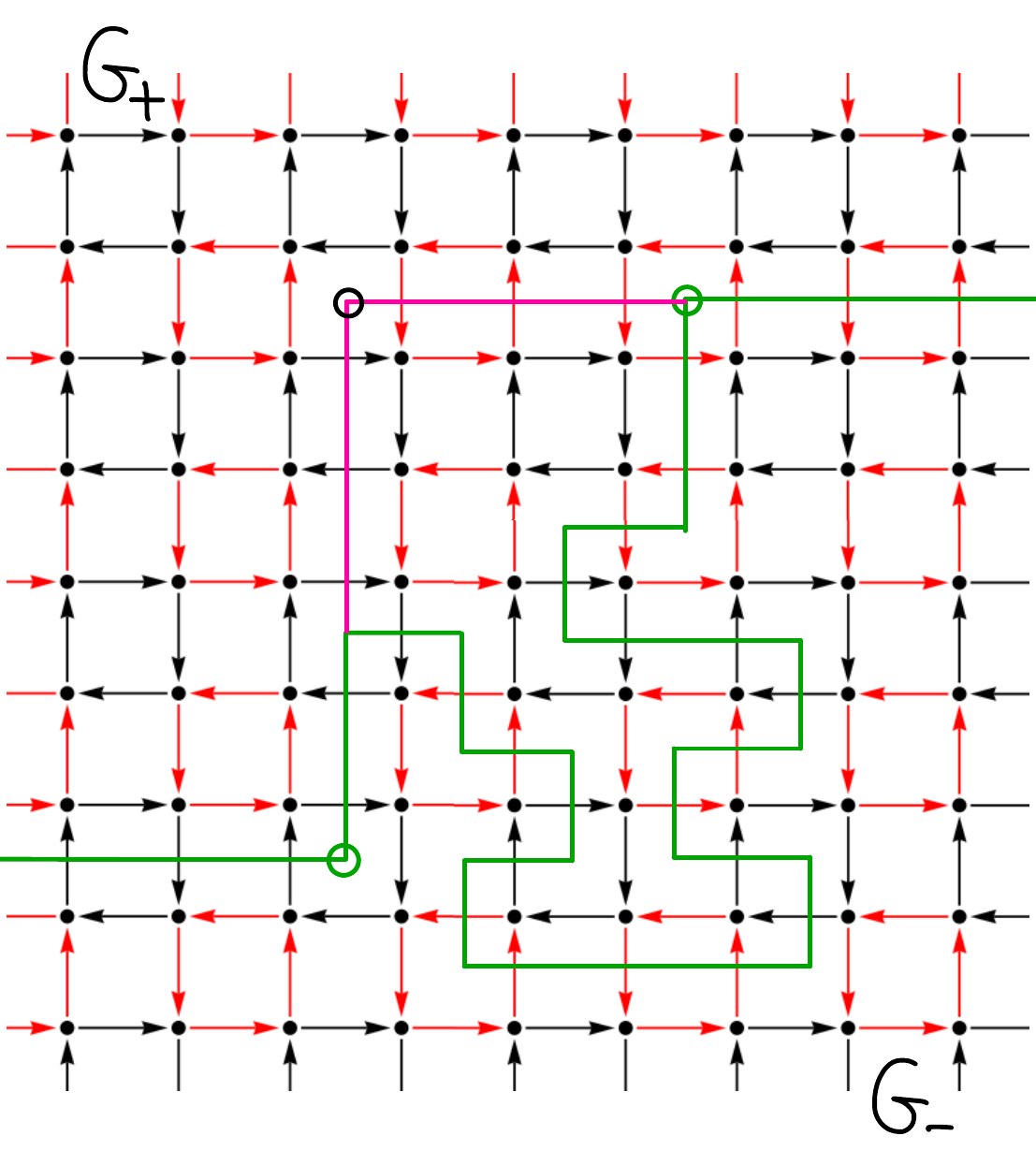}
}
\caption{Join r path to t path.}
\label{fig:joinrt}
\end{figure}

\begin{rem}\label{rem:cap}Observe that the Corollary applies to the dynamically localised case studied in \cite{ABJ2,ABJ1}, i.e.: $\vert r_z\vert=const\neq0$ sufficiently small and i.i.d. and  uniformly distributed phases. Another remarkable situation which is covered  by the Corollary  is the translation invariant $0<\vert r_z\vert=const\neq\frac{1}{\sqrt{2}}$ case with all phases equal to unity, here the spectrum is absolutely continuous and exhibits gaps.  However, as in the critical case $\vert r_z\vert=\frac{1}{\sqrt{2}}$, the flux operator $\Phi_\gamma$ is not trace class or even compact. 

To close this analysis of the transport properties of the Chalker-Coddington model, we argue that the random case is  an example of a non trivial anomalous charge transport in a regime of Anderson localisation in the following sense: consider $U$ defined by $\left\{S_{z}\right\}_{z\in\bZ\times2\bZ}$such that $\vert r_z\vert=const\neq0$ and such that $\sigma(U)$ is pure point. For $\varepsilon_0>0$ choose   an admissible path $\gamma$  like in Theorem \ref{thm:main} which is an $r$-path in $(-\infty, -\frac{1}{\varepsilon_0}\rbrack$ and a  $t$-path in $\lbrack \frac{1}{\varepsilon_0}, \infty)$. Now define for $0<\varepsilon\le\varepsilon_0$ $U_\varepsilon$ by the scattering matrices
\[S^{\varepsilon}_z:=
\left\{
\begin{array}{cl}
 \left(
\begin{array}{cc}
  0   &-1   \\
 1   &  0 
\end{array}\right) &    z\in V_{\gamma\restriction (-\infty, -\frac{1}{\varepsilon}\rbrack}   \\
 \left(
\begin{array}{cc}
  1  &0  \\
 0   & 1  
\end{array}\right) &    z\in V_{\gamma\restriction (\frac{1}{\varepsilon}, \infty\rbrack}    \\
  S_z&      z \hbox{ elsewhere }
\end{array}
\right.
\]

Let  $\Phi:=U^\ast P_\gamma U- P_\gamma$, $\Phi_\varepsilon:=U^\ast_\varepsilon P_\gamma U_\varepsilon- P_\gamma$.   $\Phi_\varepsilon$ is trace class, thus, considered as a symbol, its fermionic second quantisation is a bona fide operator in the $CAR$ observables. Furthermore for  the full quasifree state, whose symbol is the identity operator $\bI$, the mean transported charge across $\gamma$ reads: $trace(\bI\Phi_\varepsilon)=trace(\Phi_\varepsilon)=\ind(\Phi_\varepsilon)$. From  Theorem ({\ref{thm:main}}) we know that $\vert\ind(\Phi_\varepsilon)\vert=1$ so that 
\[\lim_{\varepsilon\to0}trace(\bI\Phi_\varepsilon)=\ind(\Phi)\neq0\]
and we have an anomalous charge transport by $U$ across $\gamma$.  Note that $U_\varepsilon\to U$  and $\Phi_\varepsilon\to\Phi$ strongly.

One might draw an analogy between this very concrete example and  the phenomenon  of anomalous currents in gauge field theories, see \cite{NAK}.
\end{rem}

{\bf\large  Acknowledgments} \vspace{.2cm}

We thank Claude Alain Pillet for insightful discussions. We acknowledge the warm hospitality and support of the Centre de Recherches Math\'ematiques in Montr\'eal: A.J. was a CRM-Simons Scholar-in-Residence Program, and J.A. was a guest of the Thematic Semester Mathematical Challenges in Many Body Physics and Quantum Information.

\end{document}